%% file: main.tex
\documentclass[conference]{IEEEtran}
\IEEEoverridecommandlockouts
\usepackage{cite}
\usepackage{amsmath,amssymb,amsfonts}
\usepackage{algorithmic}
\usepackage[utf8]{inputenc}
\usepackage[english]{babel}
\usepackage{amsthm}
\usepackage{bm}
\usepackage{graphicx}
\usepackage{xcolor}
\usepackage{textcomp}
\usepackage{algorithm}
\usepackage{romannum}
\def\BibTeX{{\rm B\kern-.05em{\sc i\kern-.025em b}\kern-.08em
    T\kern-.1667em\lower.7ex\hbox{E}\kern-.125emX}}

\input{macros}

\begin{document}
\title{Rate Splitting for Massive MIMO Multi-carrier system using Full Duplex Decode and Forward Relay with Hardware Impairments \\
% \thanks{This work has been funded by DFG project -CoCoMiMo-MA1184/37-1}
}

\author{\IEEEauthorblockN{Vimal Radhakrishnan, Omid Taghizadeh, and Rudolf Mathar}
\IEEEauthorblockA{\textit{Institute for Theoretical Information Technology, RWTH University Aachen,52074 Aachen, Germany} \\
\{radhakrishnan, taghizadeh, mathar\}@ti.rwth-aachen.de}
}

\IEEEpeerreviewmaketitle

\maketitle

\begin{abstract}
In this paper, we address the power allocation problem for a decode and forward relay (DF) system, where a massive multiple-input-multiple-output (mMIMO) multi-carrier (MC) base station (BS) node  communicates with a MC single antenna node directly and also through the single antenna full duplex (FD) MC relay, using rate splitting (RS) approach. Successive interference cancellation approach is adopted at the destination. We consider orthogonal frequency division multiplexing (OFDM) as our MC strategy. We take into account the impact of hardware distortions resulting in residual self-interference and inter-carrier leakage (ICL), and also imperfect channel state information (CSI). We formulate a joint sub-carrier and power allocation problem to maximize the total sum rate. An iterative optimization method is proposed, which follows successive inner approximation (SIA) framework to reach the convergence point that satisfies the Karush–Kuhn–Tucker (KKT) conditions. Numerical results show the significance of distortion-aware design for such systems, and also the significant gain in terms of sum rate compared to its half duplex (HD) and also non-rate splitting scheme.
\end{abstract}

\begin{IEEEkeywords}
Massive MIMO, Full duplex, Multi-carrier, Power allocation, Imperfect CSI, rate splitting
\end{IEEEkeywords}

\section{Introduction}
\input{Intro}

\section{System Model} \label{SM}
\input{Sys_Model_1}

\section{Optimization Problem}\label{OP}
\input{optimization}

\section{SIMULATION RESULTS}\label{SR}
 \input{Simulations}

\section{CONCLUSIONS}\label{CON}
\input{Conclusion}

\bibliographystyle{IEEEtran} 
\bibliography{main}

%\section*{References}

\end{document}

%% file: macros.tex
\usepackage{xstring}
\usepackage{dsfont}
\usepackage{xparse} 

\newcommand*{\alphabet}{abcdefghijklmnopqrstuvwxyzABCDEFGHIJKLMNOPQRST123456789}
\renewcommand*{\vec}[1]{
\IfSubStr{\alphabet}{#1}{
\ensuremath{\mbf{\MakeLowercase{#1}}}
}{
\ensuremath{\bm{\MakeLowercase{#1}}}
}
}

\newcommand*{\mat}[1]{
\IfSubStr{\alphabet}{#1}{
\ensuremath{\mbf{\MakeUppercase{#1}}}
}{
\ensuremath{\bm{\MakeUppercase{#1}}}
}
}
\newcommand*{\vh}{{\vec{h}}}

\newcommand*{\vv}{{\vec{v}}}

\newcommand*{\vx}{{\vec{x}}}

\newcommand*{\E}{\mathbb E}
\newcommand*{\K}{\mathbb K}

\newtheorem{theorem}{Theorem}[section]
\newtheorem{lemma}[theorem]{Lemma}

\newcommand*{\mbf}{\mathbf}

\newcommand*{\tr}{\mathrm{Tr}}

\NewDocumentCommand{\norm}{sm}{\IfBooleanTF{#1}{\|#2\|}{\left\|#2\right\|}}

\renewcommand*{\epsilon}{\varepsilon}
\renewcommand*{\hat}{\widehat}
\renewcommand*{\tilde}{\widetilde}
\renewcommand*{\bar}{\overline}

\NewDocumentCommand{\etal}{s}{\IfBooleanTF{#1}{\textit{et al}}{\textit{et al}. }}

\newcommand*{\diag}[1]{\text{diag}\left(#1\right)}

\newcommand*{\cnormdist}{\mathcal{CN}}
\newcommand*{\subc}{k}     %subcarrier variable
\newcommand*{\subcm}{m}    %sub carrier variable or inner sums
\newcommand*{\subcset}{\K} % subcarrier set
\newcommand*{\ted}{\text{d}}
\newcommand*{\tee}{\text{e}}
\newcommand*{\ter}{\text{r}}
\newcommand*{\tes}{\text{s}}
\newcommand*{\tet}{\text{t}}
\newcommand*{\ten}{\text{n}}

\newcommand*{\tesr}{\text{sr}}
\newcommand*{\terr}{\text{rr}}
\newcommand*{\tesd}{\text{sd}}
\newcommand*{\terd}{\text{rd}}
\newcommand*{\Sigchnlsr}{\sigma_{\tee,\tesr}} % for channel estimation noise
\newcommand*{\Sigchnlrr}{\sigma_{\tee,\terr}}
\newcommand*{\Sigchnlsd}{\sigma_{\tee,\tesd}}
\newcommand*{\Sigchnlrd}{\sigma_{\tee,\terd}}
\newcommand*{\Signr}{\sigma_{\ten,\ter}} %% for  noise 
\newcommand*{\Signd}{\sigma_{\ten,\ted}}

\newcommand*{\TxdisM}{\mathbf{0}_{N_{\mathrm{S}}}}
\newcommand*{\TxdisMat}{\mathbf{{\Theta}}_{\tet,\tes}}

%% file: Intro.tex
In 5G systems, full duplex (FD) and massive MIMO (mMIMO) are considered as two promising technologies to overcome capacity crunch and spectrum scarcity.
Simultaneous transmission and reception at the same frequency-time channel in FD systems improve the spectral efficiency compared to current half duplex wireless systems \cite{Choi10}, where the transmission and reception are separated either in time or frequency.
The main challenge for such systems is to suppress the strong interference received from its own transmitter.
Recently, some studies have been conducted in this regard \cite{ASaPScDGuDWBlSRaRWi14, DBhSKa14,GJGFHGJCoTRiRWi2019} and  various techniques \cite{TRiRWi12,EvAsAs14 , SimCKCKC16, NSmMDjAAz17} were developed in order to mitigate this self-interference.
Removing the known transmitted signal from the received ones is the key idea for self-interference cancellation.
This is challenging due to the imperfect transmitter/receiver chain components, aging of the components, imperfect knowledge of the self interference channel, etc. 

A large number of antennas are equipped in an mMIMO communication system which improve the spectral efficiency and energy efficiency using large spatial diversity as well as beamforming techniques.
On the other hand, the hardware cost for such a system becomes expensive due to the large number of antenna arrays requirements.
In order to reduce the cost, the inexpensive or less efficient transmitter/receiver chain components such as low-resolution ADC, DAC \cite{Liu2018}, low-cost power amplifiers are preferred.
More hardware distortions are introduced to the system because of these less efficient components and their aging over time.
In particular, for an FD MC system these non-linear hardware distortions lead to ICL.
Even if one of the subcarriers is employed with a high-power transmission will introduce a higher residual self-interference in all of the subcarrier channels.
So, in an FD mMIMO system, it is important to have a distortion-aware design which considers the impact of distortions due to the hardware impairments.      

FD relay has gained its attention for its improved spectral efficiency compared to the HD case due to the simultaneous transmission and reception capability \cite{HJuEOhDHo09,TRiSWeRWiJHa09,DWKNgESLoRSc12,VRaOTaRMa18,XXiDZhKXuWMaYXu15} and also reduces the overall latency of the relay communication \cite{AKaSGhASe17}.
The resource allocation problem for an FD MC system is addressed in \cite{NLiYLiMPeWWa16, YSuDWKNgZDiRSc17, ACCiKRiYRoTRa15}, where single antenna transceivers without hardware impairments are considered.
In \cite{OTaVRaACCiRMaLLa18,VRaOTaRMa18}, an alternating quadratic convex program is proposed to solve the resource allocation problem of FD MC MIMO transceiver system with hardware distortions taken in to account.
The resouce allocation problem for FD mMIMO communication system have been studied in \cite{AShKWoMReGZhKAHaJTa17, ChLYFRiJiHRoBLeViCM18,WXiXXiYXuKXuYWa17, XXiDZhKXuWMaYXu15,VRaOTaRMa2019,LLiJHeLYaZHaMPaWChHZhXLi19}. 
In \cite{WXiXXiYXuKXuYWa17, XXiDZhKXuWMaYXu15}, the resouce allocation is addresed FD massive MIMO relay with consideration of hardware distortion as well as imperfect CSI, but for a single carrier system.
For the above works on FD mMIMO relay system, the direct link is not considered. In \cite{VRaOTaRMa18}, the the direct link is considered as interference.  
In order to make use of the direct link to improve the spectral efficiency using RS can be implemented with an successive interference cancellation techniques at the receiver.
RS approaches are used in different scenarios to improve the overall system performance \cite{CHaYWuBCl15,APaTRa17,APaTRa18}.

%In \cite{AShKWoMReGZhKAHaJTa17}, For a FD multi-cell multi-user massive MIMO system, a self interference aware fractional power control mechanism is proposed where the usernodes adjust their transmit power based on the distance-dependent path-loss, SI, and maximum available transmit power.

In this paper, we investigate a downlink communication between an mMIMO transceiver base-station and an HD single antenna user node, where the BS node uses RS approach to communicate to the user through direct link as well as relay link.
The relay is a single antenna FD node and the single antenna destination node uses successive interference cancellation techniques to process the received signal.
We take into account the residual interference due to the impact of hardware distortions, ICL and imperfect CSI.
In Section \ref{SM}, we model the operation of an OFDM relay communication system and formulate the impact of imperfect CSI as well as the impact of hardware distortions.
In Section \ref{OP}, we formulate an optimization problem for joint subcarrier and power allocation to maximize the system sum-rate, which belongs into the class of smooth difference-of-convex (DC) optimization problems.
We propose an iterative optimization solution using SIA framework, which converges to a point that satisfies KKT conditions.
In Section \ref{SR}, using numerical simulations, we evaluate the performance of our proposed algorithms.
It is observed that, for high SNR scenarios, rate splitting approach performs better compared to non-rate splitting, and half duplex schemes. In section \ref{CON}, we summarize our main results. 
       
\subsection{Mathematical Notation}

Throughout this paper, we denote the vectors and matrices by lower-case and upper-case bold letters, respectively.
We use $ \mathbb{E} \{ .\}$, $\text{Tr}(.)$, $(.)^{-1}$, $ (.)^{*} $, $(.)^{T}$, and $(.)^{H}$ for mathematical expectation, trace, inverse, conjugate, transpose, and Hermitian transpose, respectively.
We use $\text{diag}(.)$ for the diag operator, which returns a diagonal matrix by setting off-diagonal elements to zero.
We denote an all-zero matrix of size $ m \times n $ by $\mathbf{0}_{ m \times n} $.
We represent the Euclidean norm as $\| . \|_2 $.
We denote the set of real, positive real, and complex numbers as $ \mathbb{R}$, $ \mathbb{R}^{+} $, and $ \mathbb{C} $ respectively. We use $|.|$ for the cardinality of a set.

%% file: Sys_Model_1.tex
We consider a downlink communication between an mMIMO FD BS with a FD single antenna user node through a direct channel and also through a relay channel.
The BS uses RS approach to communicate with the destination directly, and also using a single antenna FD relay.
In other words, two messages $s_{d,d}$ and $s_{d,r}$ destined to a user are precoded separately and  superimposed with different power levels, and are then transmitted to destination simultaneously through a direct link (DL) and relay link (RL) respectively.
In practical scenarios, the relay can be considered as an inactive user node which can be used to relay information to another user.
Thereby, the overall spectral efficiency of the system can be increased.
Since, we consider the DF relay to be FD, it can receive and transmit the signal simultaneously with a small processing delay.
In this work, we assume the processing delay to be negligible.

 Let $N_{\mathrm{BS}}$ be the number of transmit antennas at the BS node. We denote the index set of all subcarriers by $ \mathbb{K}$. 
  Furthermore, $\mathbf{h}_{sr}^\subc $ and $\mathbf{h}_{sd}^\subc \in \mathbb{C}^{1 \times N_{\mathrm{BS}}}  $ represent the $k$-th subcarrier channel from the BS to relay and destination, respectively.
  % $ h_{ii} \in \mathbb{C}^{1}$ is the self interference of the user node $i \in \mathbb{N}$. 
 The self interference channel of the relay is denoted by $h_{rr}^\subc$. 
 % $${h}_{ij}^\subc \in \mathbb{C}^{1}$ represents the co-channel interference channel from the $j$-th node to the $i$-th node, when $ i \neq j$.

In this work, we assume all the channels are constant for each frame, frequency-flat in each subcarrier and only the imperfect CSI is known. We consider similar channel error model used in \cite{WaJiPaDa09, ACCiYRoYH14}, where the true channel can be decomposed into the estimated channel and estimation error. The channel error model can be expressed as       
 \begin{equation}
 \begin{aligned}
\mathbf{h}_{\mathcal{X}}^\subc  & =   \hat{\mathbf{h}}_{\mathcal{X}}^\subc  + \tilde{\mathbf{h}}_{\mathcal{X}}^\subc , \; \; \hat{\mathbf{h}}_{\mathcal{X}}^\subc \perp \tilde{\mathbf{h}}_{\mathcal{X}}^\subc, \; \;  \forall  \mathcal{X} \in \{\tesr,\terd,\tesd, \terr\}, \forall k \in \mathbb{K}, \\
 \end{aligned}
\end{equation}
where $ \hat{\vh}_{\mathcal{X}}^\subc $ and $ \tilde{\vh}_{\mathcal{X}}^\subc $ represents the estimated channel and channel estimation error for the $\subc$ subcarrier. The entries of channel estimation error $ \tilde{\mathbf{h}}_{\mathcal{X}}^\subc $ are independent and identically distributed (i.i.d.) complex Gaussian with zero mean and variance $(\sigma_{e,\mathcal{X}}^\subc)^2$.  We assume the estimated channel and estimation error become statistically uncorrelated, for example, by considering the minimum mean square error (MMSE) channel estimation strategy at the relay and destination. 
% consider the receiver performs channel estimation which allows us to assume

The source symbol from the BS to the user using $k$-th subcarrier through DL can be represented as $s_{d,k}^{\mathrm{DL}} \in \mathbb{C}^1$, while $s_{d,k}^{\mathrm{RL}} \in \mathbb{C}^1 $ denotes the source symbol from the BS through RL using $k$-th subcarrier.
We assume the symbols are i.i.d. with unit power, i.e. $ \E \{ s_{d,k}^{\mathrm{DL}} (s_{d,k}^{\mathrm{DL}})^* \} = 1$ and $ \E \{ s_{d,k}^{\mathrm{RL}} (s_{d,k}^{\mathrm{RL}})^* \} = 1$. Let  $ \mathbf{v}_{s,r}^\subc = \tilde{\mathbf{v}}_{s,r}^\subc \sqrt{p_{sr,k}} $ and $ \mathbf{v}_{s,d}^\subc = \tilde{\mathbf{v}}_{s,d}^\subc \sqrt{p_{sd,k}} $ represent the transmit precoders at the BS for the destination and relay, respectively, $\tilde{\mathbf{v}}_{s,r}^\subc$ and  $\tilde{\mathbf{v}}_{s,d}^\subc$ denote the normalised transmit precoders.
The transmit power dedicated to the relay and destination nodes are $p_{sr,k}$ and $p_{sd,k}$, respectively.
The total available transmit power at the source and relay can be represented as $P_s$ and $P_r$, respectively.
The transmit signal from the source can be written as
\begin{equation}\label{transmitsignalsource}
\begin{aligned}
\mathbf{x}_s^\subc= \underset{: \tilde{\mathbf{x}}_s^\subc}{\underbrace{  \tilde{\mathbf{v}}_{sr}^\subc \sqrt{p_{sr,k}} (s_{r,k})+ \tilde{\mathbf{v}}_{sd}^\subc \sqrt{p_{sd,k}} (s_{d,k})}} + \mathbf{e}_{\mathrm{t},s}^ k  , \; \; \forall k \in \mathbb{K},
\end{aligned}
\end{equation}
where $e_{\mathrm{t},s}^ k$  and $\tilde{\mathbf{x}}_s^\subc$ are the transmitter distortion and the intended transmit signal at the source node, respectively. 

Correspondingly, the transmit and received signal at the relay node can be expressed as
\begin{equation}\label{signalrelay}
\begin{aligned}
{x}_r^\subc &= \underset{: \tilde{{x}}_r^\subc}{\underbrace{  \sqrt{p_{rd,k}} \hat{s}_{r,k}}} + \mathbf{e}_{\mathrm{t},r}^ k  , \; \; \forall k \in \mathbb{K}, \\
{y}_r^\subc  &=  \underset{: = \tilde{y}_r^\subc}{\underbrace{\mathbf{h}_{sr}^\subc \mathbf{x}_s^\subc +  {h}_{rr}^\subc x_r^\subc +   n_r^\subc }  } + e_{\mathrm{r},r}^ k,\; \;  \forall k \in \mathbb{K}, 
\end{aligned}
\end{equation}
where  $e_{\mathrm{t},r}^ k$ and ${e}_{\mathrm{r},r}^ k$ are the transmitter and receiver distortion at the relay node, respectively, and $n_r^\subc  \sim \cnormdist \left(0,(\sigma_{n_r}^\subc)^2  \right) $ is the noise at the relay.
The intended transmit and received signal at the relay can be represented as $\tilde{{x}}_r^\subc$ and $ \tilde{y}_r^\subc $, respectively.
The signal, which is obtained after applying SIC to the received signal (removing the known part of the transmitted signal), can be expressed as
\begin{equation}\label{recievesignal_relay}
\begin{aligned}
\bar{y}_r^\subc & = y_r^\subc - \hat{h}_{rr}^\subc \sqrt{p_{rd,k}} \hat{s}_{r,k}, \; \; \forall k \in \mathbb{K}. 
\end{aligned}
\end{equation}
 Furthermore, the received signal at the destination can be obtained as
\begin{equation}\label{signaldest}
\begin{aligned}
{y}_d^\subc  &=  \underset{: = \tilde{y}_d^\subc}{\underbrace{\mathbf{h}_{sd}^\subc \mathbf{x}_s^\subc +  {h}_{rd}^\subc x_r^\subc +   n_d^\subc }  } + e_{\mathrm{r},d}^ k,\; \;  \forall k \in \mathbb{K}, 
\end{aligned}
\end{equation}
where $e_{\mathrm{r},d}^ k$  and $\tilde{y}_d^\subc$ are the receiver distortion and the intended received signal at the destination, respectively.

Based on \cite{WNa05, GSaFMa98,HSuTVATrIBCoGDaMHe08, MDuCDiASa12}, the inaccuracy of hardware components on transmit and receive chain such as ADC and DAC error, noises caused by power amplifiers, AGC and oscillator are jointly modeled for FD MIMO transceiver in \cite{BDaAdRMaDaWBlPSc11,BPDDWBlARMaPSc12}. It has been observed that the distortion terms are proportional to the intensity of the intended signals. In this work, our MC strategy is OFDM. Therefore, we characterize the impact of these hardware distortions in the frequency domain as in \cite{OTaVRaACCiRMaLLa18}:

\begin{lemma} \label{OFDM_Distortion}
Let's define $  \tilde{x}_l^m$ and $\tilde{y}_l^m$ as the intended transmit and receive signal via $m$-th subcarrier at the $l$-th transmit/receive chain. The impact of hardware distortions in the frequency domain is characterized as  
\begin{align}\hspace{-0mm}
& {e}_{\text{t},l}^\subc  \sim \cnormdist \left( 0, \frac{{\tilde{\kappa}}_{l}}{K} \sum_{m \in \mathbb{F}_K}  \E \left\{ \left| \tilde{y}_l^m \right|^2 \right\} \right),\;\; e_{\text{t},l}^\subc \bot \tilde{y}_l^\subc, \;\; e^\subc_{\text{t},l} \bot {e^\subc_{\text{t},{l{}'}}},  \\
& {e}_{\text{r},l}^\subc  \sim \cnormdist \left( 0, \frac{\tilde{\beta}_{l}}{K} \sum_{m \in \mathbb{F}_K}  \E \left\{ \left| \tilde{x}_l^m \right|^2 \right\} \right),\;\; e_{\text{r},l}^\subc \bot \tilde{x}_l^\subc, \;\; e^\subc_{\text{r},l} \bot {e^\subc_{\text{r},{l{}'}}}, 
\end{align} 
transforming the statistical independence, as well as the proportional variance properties from the time domain. Here, $K$ represents the total number of subcarriers. The $\tilde{\kappa}_{l}$ and $\tilde{\beta}_{l}$ correspond to the transmit and receive distortion coefficientsat the $l$-th transmit/receive chain.   
\end{lemma}
\begin{proof}
Please refer to the appendix of \cite{OTaVRaACCiRMaLLa18}.
\end{proof}

In this work, we use a similar model for the transmit and receiver distortions as that of \cite{OTaVRaACCiRMaLLa18,VRaOTaRMa18}. The statistics of the distortion terms can be obtained as
 \begin{equation}\label{TXdisUSER}
\mathbf{e}_{\tet,\tes}^\subc  \sim \cnormdist \left( \TxdisM , \frac{1}{K} \mathbf{\tilde{\Theta}}_{\tet,\tes} \underset{\subc \in \subcset}{\sum}  \diag {\E \{ \tilde{\vx}_\tes^\subc (\tilde{\vx}_\tes^\subc)^H \} } \right), 
\end{equation}
\begin{equation}\label{TXdisBS}
\mathbf{e}_{\tet,\ter}^\subc  \sim \cnormdist \left(0, \frac{\tilde{\kappa}_\ter}{K}  \underset{\subc \in \subcset}{\sum}  \left( \E \{ \tilde{x}_\ter^\subc (\tilde{x}_\ter^\subc)^H \} \right) \right), 
\end{equation}
 \begin{equation}\label{RXdisUSER}
\mathbf{e}_{\ter,\ter}^\subc  \sim \cnormdist  \left( 0, \frac{\tilde{\beta}_\ter}{K}  \underset{\subc \in \subcset}{\sum} \left( \E \{ \tilde{y}_\ter^\subc (\tilde{y}_\ter^\subc)^H \} \right) \right), 
\end{equation}
\begin{equation}\label{RXdisBS}
\mathbf{e}_{\ter,\ted}^\subc  \sim \cnormdist \left( 0, \frac{\tilde{\beta}_{\ted}}{K}  \underset{\subc \in \subcset}{\sum}  \left( \E \{ \tilde{y}_\ted^\subc (\tilde{y}_\ted^\subc)^H \} \right) \right), 
\end{equation}
where $ \tilde{\kappa}_\ter , \tilde{\beta}_\ter $ represents the transmit and receiver distortion coefficient of the relay.
The receiver distortion of the destination is given by $\tilde{\kappa}_{\tes}$. The diagonal matrices $\mathbf{\tilde{\Theta}}_{\tet,\tes}$ consist of transmit distortion coefficients for the corresponding chains at the source node.
Let's define $ \kappa_{\ter} =\frac{\tilde{\kappa}_\ter}{K}$,  ${\beta}_{\ter}= \frac{\tilde{\beta}_{\ter}}{K} $, ${\beta}_{\ted}= \frac{\tilde{\beta}_{\ted}}{K} $, and $ \mathbf{{\Theta}}_{\tet,\tes}= \frac{1}{K} \mathbf{\tilde{\Theta}}_{\tet,\tes}$ for further calculations.
Please refer to the \cite[Section \Romannum{2}.A]{OTaVRaACCiRMaLLa18} for the detailed description of the used distortion model.
The above equations \eqref{TXdisUSER}, \eqref{TXdisBS}, \eqref{RXdisUSER}, and \eqref{RXdisBS} explicitly indicate the impact of the ICL, i.e., the distortion signal variance at each subcarrier is associated to the total distortion power at the corresponding chain.  

By employing Lemma \ref{OFDM_Distortion}, and equations \eqref{TXdisUSER}, \eqref{TXdisBS}, \eqref{RXdisUSER}, and \eqref{RXdisBS} on \eqref{recievesignal_relay}, the covariance of received collective interference-plus-noise signal at the relay can be formulated as
{\small
\begin{equation}\label{cov_r}
\begin{aligned}
\Sigma_\ter^\subc & \approx  \hat{\vh}_{\tesr}^\subc \tilde{\vv}_{\tesd} ^\subc  p_{\tesd}^\subc (\tilde{\vv}_{\tesd} ^\subc)^H (\hat{\vh}_{\tesr}^\subc)^H  + (\Sigchnlsr^\subc )^2 (p_{\tesd}^\subc + p_{\tesr}^\subc)  \\ 
                  & + \hat{h}_\terr^\subc \kappa_\ter \underset{\subcm \in \subcset}{\sum} p_\terd^\subcm (\hat{h}_\terr^\subc)^{*} + (\Sigchnlrr^\subc )^2 \left( \kappa_\ter \underset{\subcm \in \subcset}{\sum} p_\terd^\subcm + p_\terd^\subc \right) \\
                  & + \hat{\vh}_\tesr^\subc \mathbf{{\Theta}}_{\tet,\tes} \underset{\subcm \in \subcset}{\sum} \bigg(  \diag{ \tilde{\vv}_\tesr^\subcm p_{\tesr}^\subcm (\tilde{\vv}_\tesr^\subcm)^H} 
                   + \diag{\tilde{\vv}_\tesd^\subcm p_{\tesd}^\subcm (\tilde{\vv}_\tesd^\subcm)^H} \bigg) \\ 
                  & (\hat{\vh}_\tesr^\subc)^H + (\Sigchnlsr^\subc )^2 \tr \bigg( \mathbf{{\Theta}}_{\tet,\tes} \underset{\subcm \in \subcset}{\sum} \bigg(  \diag{ \tilde{\vv}_\tesr^\subcm p_{\tesr}^\subcm (\tilde{\vv}_\tesr^\subcm)^H} \\ 
                  &+ \diag{\tilde{\vv}_\tesd^\subcm p_{\tesd}^\subcm (\tilde{\vv}_\tesd^\subcm)^H} \bigg) \bigg) + \beta_\ter \underset{\subcm \in \subcset}{\sum} \bigg( \hat{\vh}_{\tesr}^\subcm \tilde{\vv}_{\tesr} ^\subcm  p_{\tesr}^\subcm (\tilde{\vv}_{\tesr} ^\subcm)^H (\hat{\vh}_{\tesr}^\subcm)^H \\ 
                  & +  \hat{\vh}_{\tesr}^\subcm \tilde{\vv}_{\tesd} ^\subcm  p_{\tesd}^\subcm (\tilde{\vv}_{\tesd} ^\subcm)^H (\hat{\vh}_{\tesr}^\subcm)^H + (\Sigchnlsr^\subcm )^2  \tr \bigg(  \tilde{\vv}_{\tesr} ^\subcm  p_{\tesr}^\subcm (\tilde{\vv}_{\tesr} ^\subcm)^H   \\ 
                  & +   \tilde{\vv}_{\tesd} ^\subcm  p_{\tesd}^\subcm (\tilde{\vv}_{\tesd} ^\subcm)^H  \bigg)  + \hat{h}_\terr^\subcm p_\terd^\subcm (\hat{h}_\terr^\subcm)^{*} + (\Sigchnlrr^\subcm )^2   p_\terd^\subcm + (\Signr^\subcm)^2 \bigg) \\
                  & + (\Signr^\subc )^2.
\end{aligned}
\end{equation}
}
Since the transmit and receive distortion coefficients $\tilde{\kappa}$ and $ \tilde{\beta} $ lie within the range of $0$ and $1$ and mostly have very small values, the higher order terms of the transmit and receive distortion are ignored.
In the above equation \eqref{cov_r}, the first two terms correspond to the co-channel interference.
The next two terms relate to the self-interference at the relay.
The fifth and sixth terms correspond to the transmit distortion at the source.
The next term represents the receiver distortion at the relay.
The last term is the thermal noise at the relay.
 
At the destination node, the received signal is processed in two phase using successive interference cancellation technique.
In the first phase, the received signal from the source is considered as interference while considering the strong received signal from the relay as desired signal.
In the second phase, as the received signal from the relay becomes known, this part can be removed from the received signal which in turn reduces the total interference. The covariance of received collective interference-plus-noise signal at the destination for the first phase can be expressed as
{\small
\begin{equation}\label{cov_d}
\begin{aligned}
\Sigma_{\ted,\mathrm{1}}^\subc & \approx  \hat{\vh}_{\tesd}^\subc \left( \tilde{\vv}_{\tesr}^\subc p_{\tesr}^\subc (\tilde{\vv}_{\tesr} ^\subc)^H + \tilde{\vv}_{\tesd}^\subc p_{\tesd}^\subc (\tilde{\vv}_{\tesd} ^\subc)^H \right) (\hat{\vh}_{\tesd}^\subc)^H   \\ 
				  & + (\Sigchnlsd^\subc )^2 \left( \tr \left( \tilde{\vv}_{\tesr}^\subc p_{\tesr}^\subc (\tilde{\vv}_{\tesr} ^\subc)^H + \tilde{\vv}_{\tesd}^\subc p_{\tesd}^\subc (\tilde{\vv}_{\tesd} ^\subc)^H \right) \right)  \\ 
				  & + \hat{h}_\terd^\subc \kappa_\ter \underset{\subcm \in \subcset}{\sum} p_\terd^\subcm (\hat{h}_\terd^\subc)^{*} + (\Sigchnlrd^\subc )^2 \left( \kappa_\ter \underset{\subcm \in \subcset}{\sum} p_\terd^\subcm + p_\terd^\subc \right) \\
                  & + \hat{\vh}_\tesd^\subc \mathbf{\Theta}_{\tet,\tes} \underset{\subcm \in \subcset}{\sum} \bigg( \diag{ \tilde{\vv}_\tesr^\subcm p_{\tesr}^\subcm (\tilde{\vv}_\tesr^\subcm)^H} 
                   + \diag{\tilde{\vv}_\tesd^\subcm p_{\tesd}^\subcm (\tilde{\vv}_\tesd^\subcm)^H} \bigg) \\ 
                  & (\hat{\vh}_\tesd^\subc)^H + (\Sigchnlsd^\subc )^2 \tr \bigg( \mathbf{{\Theta}}_{\tet,\tes} \underset{\subcm \in \subcset}{\sum} \bigg(  \diag{ \tilde{\vv}_\tesr^\subcm p_{\tesr}^\subcm (\tilde{\vv}_\tesr^\subcm)^H} \\ 
                  &+ \diag{\tilde{\vv}_\tesd^\subcm p_{\tesd}^\subcm (\tilde{\vv}_\tesd^\subcm)^H} \bigg) \bigg) + \beta_\ted \underset{\subcm \in \subcset}{\sum} \bigg( \hat{\vh}_{\tesd}^\subcm \tilde{\vv}_{\tesr} ^\subcm  p_{\tesr}^\subcm (\tilde{\vv}_{\tesr} ^\subcm)^H (\hat{\vh}_{\tesd}^\subcm)^H \\ 
                  & +  \hat{\vh}_{\tesd}^\subcm \tilde{\vv}_{\tesd} ^\subcm  p_{\tesd}^\subcm (\tilde{\vv}_{\tesd} ^\subcm)^H (\hat{\vh}_{\tesd}^\subcm)^H + (\Sigchnlsd^\subcm )^2  \tr \bigg(  \tilde{\vv}_{\tesr} ^\subcm  p_{\tesr}^\subcm (\tilde{\vv}_{\tesr} ^\subcm)^H   \\ 
                  & +   \tilde{\vv}_{\tesd} ^\subcm  p_{\tesd}^\subcm (\tilde{\vv}_{\tesd} ^\subcm)^H  \bigg)  + \hat{h}_\terd^\subcm p_\terd^\subcm (\hat{h}_\terd^\subcm)^{*} + (\Sigchnlrd^\subcm )^2   p_\terd^\subcm + (\Signd^\subcm)^2 \bigg) \\
                  & + (\Signd^\subc )^2.
\end{aligned}
\end{equation}
}
Here the first two terms in the above equation \eqref{cov_d}, represent the co-channel interference. The next two terms corresponds the the relay transmit distortion at the destination. The third and forth terms represent the source transmit distortion at the destination. The next term correspond to the receive distortion at the destination. The last term is the thermal noise at the destination. 

For the second phase, the signal from the relay is known and it can be removed from the received signal. The signal from the source to destination becomes the desired signal. The covariance of received collective interference-plus-noise signal at the destination for the second phase can be obtained as 
\begin{equation}
\begin{aligned}
\Sigma_{\ted,\mathrm{2}}^\subc := \Sigma_{\ted,\mathrm{1}}^\subc  - \hat{\vh}_{\tesd}^\subc \left( \tilde{\vv}_{\tesd}^\subc p_{\tesd}^\subc (\tilde{\vv}_{\tesd} ^\subc)^H \right) (\hat{\vh}_{\tesd}^\subc)^H.
\end{aligned}
\end{equation}

\subsection{Achievable information rate}
The achievable information rate for the source to relay link using subcarrier $k$ can be obtained as
\begin{equation}
R_{\tesr}^{\subc}=\hspace{-1mm}  \gamma_0 \mathrm{log}_2 \hspace{-1mm}\left( \hspace{-1mm} 1 + \frac{|\hat{\vh}_{\tesr}^\subc \tilde{\vv}_{\tesr}^\subc |^2 p_{\tesr}^{\subc}}{  \alpha_{\ter}^\subc + \underset{m \in \mathbb{K}}{\sum}   \big( \gamma_{\tesr}^{\subc\subcm} p_{\tesr}^{\subcm} + \gamma_{\terd}^{\subc\subcm} p_{\terd}^{\subcm} +   \gamma_{\tesd}^{\subc\subcm} p_{\tesd}^{\subcm} \big)} \right)
\end{equation}
where
\begin{equation}
\begin{aligned}
\gamma_{\tesr}^{\subc\subcm}  = & \delta_{km} (\Sigchnlsr^\subcm)^2 + \hat{\vh}_\tesr^\subc \TxdisMat \diag{ \tilde{\vv}_\tesr^\subcm (\tilde{\vv}_\tesr^\subcm)^H} (\hat{\vh}_\tesr^\subc)^H  \\
								& + (\Sigchnlsr^\subc )^2 \tr \bigg( \TxdisMat  \diag{ \tilde{\vv}_\tesr^\subcm (\tilde{\vv}_\tesr^\subcm)^H}\bigg)  \\ 
								& + \beta_\ter \bigg( \hat{\vh}_{\tesr}^\subcm \tilde{\vv}_{\tesr} ^\subcm  p_{\tesr}^\subcm (\tilde{\vv}_{\tesr} ^\subcm)^H (\hat{\vh}_{\tesr}^\subcm)^H +   (\Sigchnlsr^\subcm)^2 \bigg) ,
\end{aligned}
\end{equation}
\begin{equation}
\begin{aligned}
\gamma_{\terd}^{\subc\subcm}  = & \delta_{km} (\Sigchnlrr^\subcm)^2 +  \kappa_\ter  \big( \hat{h}_\terr^\subc (\hat{h}_\terr^\subc)^* + (\Sigchnlrr^\subc)^2 \big) \\
                                & + \beta_\ter \big( \hat{h}_\terr^\subc (\hat{h}_\terr^\subcm)^* + (\Sigchnlrr^\subcm)^2 \big) ,
\end{aligned}
\end{equation}
\begin{equation}
\begin{aligned}
\gamma_{\tesd}^{\subc\subcm}  = &  \delta_{km} \bigg( \hat{\vh}_{\tesr}^\subcm \tilde{\vv}_{\tesd} ^\subcm (\tilde{\vv}_{\tesd} ^\subcm)^H (\hat{\vh}_{\tesr}^\subcm)^H  + (\Sigchnlsr^\subcm )^2 \bigg) \\ 
                  & + \hat{\vh}_\tesr^\subc \mathbf{{\Theta}}_{\tet,\tes} \diag{\tilde{\vv}_\tesd^\subcm (\tilde{\vv}_\tesd^\subcm)^H} (\hat{\vh}_\tesr^\subc)^H  \\ 
                  & + (\Sigchnlsr^\subc )^2 \tr \bigg( \mathbf{{\Theta}}_{\tet,\tes} \diag{\tilde{\vv}_\tesd^\subcm (\tilde{\vv}_\tesd^\subcm)^H} \bigg) \\ 
                  & + \beta_\ter \bigg(  \hat{\vh}_{\tesr}^\subcm \tilde{\vv}_{\tesd} ^\subcm  p_{\tesd}^\subcm (\tilde{\vv}_{\tesd} ^\subcm)^H (\hat{\vh}_{\tesr}^\subcm)^H + (\Sigchnlsr^\subcm )^2 \bigg) ,
\end{aligned}
\end{equation}
\begin{equation}
\begin{aligned}
\alpha_{\ter}^\subc  =  (\Signr^\subc)^2 + \beta_\ter  \underset{\subcm \in \subcset}{\sum} (\Signr^\subcm)^2 ,
\end{aligned}
\end{equation}
and $\delta_{km} = 1$ when $k=m$ and otherwise $\delta_{km} =  0$.

Similarly, the achievable information rate for the relay to destination link using subcarrier $\subc$ can be expressed as

\begin{equation}
R_{\terd}^{\subc} =\hspace{-1mm}  \gamma_0 \mathrm{log}_2 \hspace{-1mm}\left( \hspace{-1mm} 1 + \frac{|\hat{h}_{\terd}^\subc |^2 p_{\terd}^{\subc}}{  \alpha_{\ted}^\subc + \underset{m \in \mathbb{K}}{\sum}   \big( \bar{\gamma}_{\tesr}^{\subc\subcm} p_{\tesr}^{\subcm} + \bar{\gamma}_{\terd}^{\subc\subcm} p_{\terd}^{\subcm} +   \bar{\gamma}_{\tesd}^{\subc\subcm} p_{\tesd}^{\subcm} \big)} \right)
\end{equation}
where
\begin{equation}
\begin{aligned}
\bar{\gamma}_{\tesr}^{\subc\subcm}  = & \delta_{\subc\subcm} \bigg(\hat{\vh}_{\tesd}^\subcm \tilde{\vv}_{\tesr}^\subcm (\tilde{\vv}_{\tesr} ^\subcm)^H (\hat{\vh}_{\tesd}^\subcm)^H +(\Sigchnlsd^\subcm )^2 \bigg)  \\
                                      & + \hat{\vh}_\tesd^\subc \mathbf{\Theta}_{\tet,\tes} \diag{ \tilde{\vv}_\tesr^\subcm (\tilde{\vv}_\tesr^\subcm)^H}  (\hat{\vh}_\tesd^\subc)^H \\ 
                                      & + (\Sigchnlsd^\subc )^2 \tr \bigg( \mathbf{{\Theta}}_{\tet,\tes}  \diag{ \tilde{\vv}_\tesr^\subcm (\tilde{\vv}_\tesr^\subcm)^H}  \bigg)  \\  
                                      & + \beta_\ted \bigg( \hat{\vh}_{\tesd}^\subcm \tilde{\vv}_{\tesr} ^\subcm  (\tilde{\vv}_{\tesr} ^\subcm)^H (\hat{\vh}_{\tesd}^\subcm)^H + (\Sigchnlsd^\subcm )^2   \bigg)
\end{aligned}
\end{equation}
\begin{equation}
\begin{aligned}
\bar{\gamma}_{\terd}^{\subc\subcm}  =  &  \delta_{\subc\subcm}(\Sigchnlrd^\subcm )^2 + \kappa_\ter \left( \hat{h}_\terd^\subc (\hat{h}_\terd^\subc)^{*} + (\Sigchnlrd^\subc )^2  \right) \\
                  & + \beta_\ted  \bigg( \hat{h}_\terd^\subcm (\hat{h}_\terd^\subcm)^{*} + (\Sigchnlrd^\subcm )^2   \bigg) 
\end{aligned}
\end{equation}
\begin{equation}
\begin{aligned}
\bar{\gamma}_{\tesd}^{\subc\subcm}  = & \delta_{\subc\subcm} \left(\hat{\vh}_{\tesd}^\subcm \tilde{\vv}_{\tesd}^\subcm (\tilde{\vv}_{\tesd}^\subcm)^H (\hat{\vh}_{\tesd}^\subcm)^H + (\Sigchnlsd^\subcm)^2 \right)\\
                                      & + \hat{\vh}_\tesd^\subc \mathbf{\Theta}_{\tet,\tes} \diag{\tilde{\vv}_\tesd^\subcm  (\tilde{\vv}_\tesd^\subcm)^H} (\hat{\vh}_\tesd^\subc)^H  \\ 
                                      & + (\Sigchnlsd^\subc )^2 \tr \bigg( \mathbf{{\Theta}}_{\tet,\tes}  \diag{\tilde{\vv}_\tesd^\subcm (\tilde{\vv}_\tesd^\subcm)^H} \bigg) \\ 
                                      & + \beta_\ted  \bigg(  \hat{\vh}_{\tesd}^\subcm \tilde{\vv}_{\tesd} ^\subcm  (\tilde{\vv}_{\tesd} ^\subcm)^H (\hat{\vh}_{\tesd}^\subcm)^H + (\Sigchnlsd^\subcm )^2    \bigg)  
\end{aligned}
\end{equation}
\begin{equation}
\begin{aligned}
\alpha_{\ter}^\subc  =  (\Signr^\subc)^2 + \beta_\ter  \underset{\subcm \in \subcset}{\sum} (\Signr^\subcm)^2 ,
\end{aligned}
\end{equation}
Finally, the achievable information rate for the source to destination link using subcarrier $\subc$ can be formulated as
\begin{equation}
R_{\tesd}^{\subc} =\hspace{-1mm}  \gamma_0 \mathrm{log}_2 \hspace{-1mm}\left( \hspace{-1mm} 1 + \frac{|\hat{\vh}_{\tesd}^\subc \tilde{\vv}_{\tesd}^\subc  |^2 p_{\tesd}^{\subc}}{  \alpha_{\ted}^\subc + \underset{m \in \mathbb{K}}{\sum}   \big( \bar{\gamma}_{\tesr}^{\subc\subcm} p_{\tesr}^{\subcm} + \bar{\gamma}_{\terd}^{\subc\subcm} p_{\terd}^{\subcm} +   \tilde{\gamma}_{\tesd}^{\subc\subcm} p_{\tesd}^{\subcm} \big)} \right)
\end{equation}
where
\begin{equation}
\begin{aligned}
\tilde{\gamma}_{\tesd}^{\subc\subcm}  = &   \bar{\gamma}_{\tesd}^{\subc\subcm} - \delta_{km} \hat{\vh}_{\tesd}^\subcm \tilde{\vv}_{\tesd} ^\subcm (\tilde{\vv}_{\tesd} ^\subcm)^H (\hat{\vh}_{\tesd}^\subcm)^H  .
\end{aligned}
\end{equation}
In this work, we consider that the source node or BS has a large antenna array.
Therefore, different well-studied linear precoder-decoder filtering strategies, such as, maximum ratio transmission/maximum ratio combining (MRT/MRC), zero forcing (ZF), MMSE and so on, are available for the selection of transmit precoders and receive decoders at the source.
We can also reduce some computational complexity to obtain the achievable rates if we assume some common assumptions from mMIMO studies on the channel covariance matrices such as Hermitian, Teoplitz, etc., as discussed in \cite{XXiDZhKXuWMaYXu15}.
Now, the total achievable information rate for the system can be  written as
\begin{equation}
R^{\subc}=R_{\tesd}^{\subc}+ \min \{ R_{\tesr}^{\subc},R_{\terd}^{\subc} \}.
\end{equation}

%% file: optimization.tex
In this section, we present the joint sub-carrier and power allocation optimization problem to maximize spectral efficiency in terms of total sum-rate under transmit power constraints. The node is not transmitting or receiving in particular sub-carrier if the power allocated to a particular sub-carrier is zero, thereby in-cooperating the sub-carrier allocation into the power allocation problem.

\subsection{Sum Rate Maximization}
The sum rate maximization problem for our system can be expressed as 
\begin{equation}
\begin{aligned}
\underset{\underset{p_{\terd}^{\subc}> 0}{p_{\tesd}^{\subc}> 0, p_{\tesr}^{\subc}> 0}}{\mathrm{max}} & \; \;  \underset{\subc \in \subcset}{\sum} R^\subc \\   
 \mathrm{subject \; to} \quad & \underset{\subc \in \subcset}{\sum} p_{\terd}^{\subc}  \leq P_{\ter},  \quad  \underset{\subc \in \subcset}{\sum} p_{\tesd}^{\subc} + p_{\tesr}^{\subc} \leq P_{\tes},
\end{aligned}
\end{equation}
where $P_{\tes}$ and $P_{\ter}$ are the available transmit power at the source node and the relay, respectively. The above optimization problem can be rewritten as
\begin{equation} \label{opt_rate}
\begin{aligned}
\underset{\underset{p_{\terd}^{\subc}> 0}{p_{\tesd}^{\subc}> 0, p_{\tesr}^{\subc}> 0}}{\mathrm{max}} & \; \;  \underset{\subc \in \subcset}{\sum} R_{\tesd}^{\subc} + t \\   
 \mathrm{subject \; to} \quad & R_{\tesr}^{\subc} \geq t , \quad R_{\terd}^{\subc} \geq t, \\
& \underset{\subc \in \subcset}{\sum} p_{\terd}^{\subc}  \leq P_{\ter}, \quad \underset{\subc \in \subcset}{\sum} p_{\tesd}^{\subc} + p_{\tesr}^{\subc} \leq P_{\tes}.
\end{aligned}
\end{equation}
The above optimization problem (\ref{opt_rate}) belongs to the class of smooth difference-of-convex optimization problems.
We propose an iterative algorithm which reaches a converging point that satisfies the KKT optimality conditions using the SIA framework \cite{SIAMBaWrGoP78}. 

Now, we use Taylor's approximation on the concave terms of rate to obtain a lower bound.
For the approximation, we first select $p_{\terd,0}^{\subc}, p_{\tesd,0}^{\subc}$ and $ p_{\tesr,0}^{\subc}$ as a feasible transmit power value for the relay-destination, source-relay, and source-destination link, respectively. 
%\min \{ R_{\tesr}^{\subc},R_{\terd}^{\subc}
A lower-bound of $ R_{\tesr}^{\subc}$, after applying Taylor's approximation on the concave terms, can be expressed as 
{\small
\begin{equation}\label{eq_R_sr}
\begin{aligned}
& R_{\tesr}^{\subc}  \geq  \gamma_0 \mathrm{log}_2 \bigg( \hspace{-1mm}  \alpha_{\ter}^\subc + \hspace{-1mm}  \underset{m \in \mathbb{K}}{\sum} \hspace{-1mm}    \big( \bar{\gamma}_{\tesr}^{\subc\subcm} p_{\tesr}^{\subcm} + \hspace{-1mm}   {\gamma}_{\terd}^{\subc\subcm} p_{\terd}^{\subcm} + \hspace{-1mm}    {\gamma}_{\tesd}^{\subc\subcm} p_{\tesd}^{\subcm} + \hspace{-1mm}  |\hat{\vh}_{\tesr}^\subc  \tilde{\vv}_{\tesr}^\subc  |^2 p_{\tesr}^{\subc} \big) \hspace{-1mm}  \bigg) \\ 
& - \gamma_0 \mathrm{log}_2 \bigg( \alpha_{\ter}^\subc + \underset{m \in \mathbb{K}}{\sum}   \big(  {\gamma}_{\tesr}^{\subc\subcm} p_{\tesr,0}^{\subcm} +  {\gamma}_{\terd}^{\subc\subcm} p_{\terd,0}^{\subcm} +    {\gamma}_{\tesd}^{\subc\subcm} p_{\tesd,0}^{\subcm} \big) \bigg )  \\ 
& -\hspace{-1mm} \frac{\gamma_0 \underset{m \in \mathbb{K}}{\sum} \hspace{-0.85mm}  \bigg(  {\gamma}_{\tesr}^{\subc\subcm} \hspace{-1mm} \left( \hspace{-0.25mm}  p_{\tesr}^{\subcm} - p_{\tesr,0}^{\subcm} \right) \hspace{-1mm} + \hspace{-1mm}  {\gamma}_{\terd}^{\subc\subcm} \hspace{-1mm} \left(  p_{\terd}^{\subcm} - p_{\terd,0}^{\subcm}\right)\hspace{-1mm} +\hspace{-1mm}    {\gamma}_{\tesd}^{\subc\subcm} \left( p_{\tesd}^{\subcm} - p_{\tesd,0}^{\subcm}\right) \bigg)}{  \mathrm{log}(2) \bigg( \alpha_{\ter}^\subc + \underset{m \in \mathbb{K}}{\sum}   \big(  {\gamma}_{\tesr}^{\subc\subcm} p_{\tesr}^{\subcm} +  {\gamma}_{\terd}^{\subc\subcm} p_{\terd}^{\subcm} +    {\gamma}_{\tesd}^{\subc\subcm} p_{\tesd}^{\subcm} \big) \bigg) }   \\ & =: \bar{R}_{\tesr}^{\subc}. 
\end{aligned}
\end{equation}
}Similarly, after applying Taylor's approximation, the lower bound of $R_{\tesd}^{\subc}$ and $R_{\terd}^{\subc}$ can be obtained as
{\small
\begin{equation}\label{eq_R_sd}
\begin{aligned}
& R_{\tesd}^{\subc}  \geq  \gamma_0 \mathrm{log}_2 \bigg( \hspace{-1mm}  \alpha_{\ted}^\subc + \hspace{-1mm}  \underset{m \in \mathbb{K}}{\sum} \hspace{-1mm}    \big( \bar{\gamma}_{\tesr}^{\subc\subcm} p_{\tesr}^{\subcm} + \hspace{-1mm}  \bar{\gamma}_{\terd}^{\subc\subcm} p_{\terd}^{\subcm} + \hspace{-1mm}   \tilde{\gamma}_{\tesd}^{\subc\subcm} p_{\tesd}^{\subcm} + \hspace{-1mm}  |\hat{\vh}_{\tesd}^\subc \tilde{\vv}_{\tesd}^\subc  |^2 p_{\tesd}^{\subc} \big) \hspace{-1mm}  \bigg) \\ 
& - \gamma_0 \mathrm{log}_2 \bigg( \alpha_{\ted}^\subc + \underset{m \in \mathbb{K}}{\sum}   \big( \bar{\gamma}_{\tesr}^{\subc\subcm} p_{\tesr,0}^{\subcm} + \bar{\gamma}_{\terd}^{\subc\subcm} p_{\terd,0}^{\subcm} +   \tilde{\gamma}_{\tesd}^{\subc\subcm} p_{\tesd,0}^{\subcm} \big) \bigg )  \\ 
& -\hspace{-1mm} \frac{\gamma_0 \underset{m \in \mathbb{K}}{\sum} \hspace{-0.85mm}  \bigg( \bar{\gamma}_{\tesr}^{\subc\subcm} \hspace{-1mm} \left( \hspace{-0.25mm}  p_{\tesr}^{\subcm} - p_{\tesr,0}^{\subcm} \right) \hspace{-1mm} + \hspace{-1mm} \bar{\gamma}_{\terd}^{\subc\subcm} \hspace{-1mm} \left(  p_{\terd}^{\subcm} - p_{\terd,0}^{\subcm}\right)\hspace{-1mm} +\hspace{-1mm}   \tilde{\gamma}_{\tesd}^{\subc\subcm} \left( p_{\tesd}^{\subcm} - p_{\tesd,0}^{\subcm}\right) \bigg)}{  \mathrm{log}(2) \bigg( \alpha_{\ted}^\subc + \underset{m \in \mathbb{K}}{\sum}   \big( \bar{\gamma}_{\tesr}^{\subc\subcm} p_{\tesr}^{\subcm} + \bar{\gamma}_{\terd}^{\subc\subcm} p_{\terd}^{\subcm} + \hspace{-1mm}  \tilde{\gamma}_{\tesd}^{\subc\subcm} p_{\tesd}^{\subcm} \big) \bigg) }   \\ & =: \bar{R}_{\tesd}^{\subc} 
\end{aligned}
\end{equation}
}
and
{\small
\begin{equation}\label{eq_R_rd}
\begin{aligned}
& R_{\terd}^{\subc}  \geq  \gamma_0 \mathrm{log}_2 \bigg( \hspace{-1mm}  \alpha_{\ted}^\subc + \hspace{-1mm}  \underset{m \in \mathbb{K}}{\sum} \hspace{-1mm}    \big( \bar{\gamma}_{\tesr}^{\subc\subcm} p_{\tesr}^{\subcm} + \hspace{-1mm}  \bar{\gamma}_{\terd}^{\subc\subcm} p_{\terd}^{\subcm} + \hspace{-1mm}   \bar {\gamma}_{\tesd}^{\subc\subcm} p_{\tesd}^{\subcm} + \hspace{-1mm}  |\hat{h}_{\terd}^\subc |^2 p_{\terd}^{\subc} \big) \hspace{-1mm}  \bigg) \\ 
& - \gamma_0 \mathrm{log}_2 \bigg( \alpha_{\ted}^\subc + \underset{m \in \mathbb{K}}{\sum}   \big( \bar{\gamma}_{\tesr}^{\subc\subcm} p_{\tesr,0}^{\subcm} + \bar{\gamma}_{\terd}^{\subc\subcm} p_{\terd,0}^{\subcm} +   \bar {\gamma}_{\tesd}^{\subc\subcm} p_{\tesd,0}^{\subcm} \big) \bigg )  \\ 
& -\hspace{-1mm} \frac{\gamma_0 \underset{m \in \mathbb{K}}{\sum} \hspace{-0.85mm}  \bigg( \bar{\gamma}_{\tesr}^{\subc\subcm} \hspace{-1mm} \left( \hspace{-0.25mm}  p_{\tesr}^{\subcm} - p_{\tesr,0}^{\subcm} \right) \hspace{-1mm} + \hspace{-1mm} \bar{\gamma}_{\terd}^{\subc\subcm} \hspace{-1mm} \left(  p_{\terd}^{\subcm} - p_{\terd,0}^{\subcm}\right)\hspace{-1mm} +\hspace{-1mm}   \bar {\gamma}_{\tesd}^{\subc\subcm} \left( p_{\tesd}^{\subcm} - p_{\tesd,0}^{\subcm}\right) \bigg)}{  \mathrm{log}(2) \bigg( \alpha_{\ted}^\subc + \underset{m \in \mathbb{K}}{\sum}   \big( \bar{\gamma}_{\tesr}^{\subc\subcm} p_{\tesr}^{\subcm} + \bar{\gamma}_{\terd}^{\subc\subcm} p_{\terd}^{\subcm} +   \bar {\gamma}_{\tesd}^{\subc\subcm} p_{\tesd}^{\subcm} \big) \bigg) }   \\ & =: \bar{R}_{\terd}^{\subc}, 
\end{aligned}
\end{equation}
}respectively.
Using this approximation, we can rewrite the optimization problem as
\begin{equation} \label{opt_rate_lb}
\begin{aligned}
\underset{\underset{p_{\terd}^{\subc}> 0}{p_{\tesd}^{\subc}> 0, p_{\tesr}^{\subc}> 0}}{\mathrm{max}} & \; \;  \underset{\subc \in \subcset}{\sum} \bar{R}_{\tesd}^{\subc} + t \\   
 \mathrm{subject \; to} \quad & \bar{R}_{\tesr}^{\subc} \geq t , \quad \bar{R}_{\terd}^{\subc} \geq t, \\
& \underset{\subc \in \subcset}{\sum} p_{\terd}^{\subc}  \leq P_{\ter}, \quad \underset{\subc \in \subcset}{\sum} p_{\tesd}^{\subc} + p_{\tesr}^{\subc} \leq P_{\tes}.
\end{aligned}
\end{equation}
Here, the objective of the above convex optimization problem $ \bar{R}^{\subc} :=  \bar{R}_{\tesd}^{\subc} + t$ is a jointly concave function over $p_{\terd}^{\subc}, p_{\tesd}^{\subc}$ and $ p_{\tesr}^{\subc}$. We propose an iterative algorithm, where for each iterative update, we now solve the above convex optimization problem to optimality.
The iterative update is continued until a stable point is reached.
Furthermore, we use a first order Taylor approximation on a smooth convex function, we can conclude that  $ \bar{R}^{\subc}$ represents a global and tight lower bound to $ {R}^{\subc}$, with a shared slope at the point of approximation \cite{BStVLi04}.
The solution can achieve a convergence point that satisfies KKT conditions since the proposed iterative update also fulfills the requirements set in \cite[Theorem 1]{SIAMBaWrGoP78}. The algorithm \ref{Alg_1} provides a detailed procedure of the algorithm.
\begin{algorithm}
\caption{$\text{\small{For sum rate maximization}}$}
\label{Alg_1}
 \begin{algorithmic}[1]
\STATE {$a \gets 0$ (set iteration number to zero)}
\STATE {$p_{\terd,0}^{\subc}, p_{\tesd,0}^{\subc},  p_{\tesr,0}^{\subc}  \gets \text{\small{uniform (equal) power  initialization}}$ }
\REPEAT
 \STATE{$a \gets a + 1$}
 \STATE{$p_{\terd}^{\subc}, p_{\tesd}^{\subc},  p_{\tesr}^{\subc} \gets  \text{\small{solve}} $ \eqref{opt_rate_lb} }
 \STATE{$p_{\terd,0}^{\subc}, p_{\tesd,0}^{\subc},  p_{\tesr,0}^{\subc}  \gets p_{\terd,0}^{\subc}\hspace{-1mm} = \hspace{-1mm} p_{\terd}^{\subc}, p_{\tesd,0}^{\subc}\hspace{-1mm} = \hspace{-1mm} p_{\tesd}^{\subc} $  $\text{\small{and}}$ $ p_{\tesr,0}^{\subc}\hspace{-1mm} = \hspace{-1mm} p_{\tesr}^{\subc} $, $\text{\small{respectively}}$ }
 \UNTIL{\text{\small{a stable point, or maximum number of $a$ reached}}}
 \RETURN { $\{  p_{\terd}^{\subc}, p_{\tesd}^{\subc},  p_{\tesr}^{\subc}  \}$ } 
 \end{algorithmic}
\end{algorithm}

%% file: Simulations.tex
By using numerical simulations, we evaluate the performance of the proposed algorithms (RS) introduced in Section \ref{OP} and compare with other benchmarks/scenarios, such as the no-distortion (RS-ND) algorithm, where the hardware distortions are not considered, Only Direct link (ODL) is present, only relay link (ORL) is available, and half-duplex (HD) algorithm.
For the simulations, we consider the MRT/MRC strategy for our transmit precoder and receive filters at the BS.
All communication channels follow an uncorrelated Rayleigh flat fading model.
The self-interference channel follows the characterization reported in \cite{MDuCDiASa12}, i.e., ${h}_{rr} \sim  \mathcal{CN} \left(\sqrt{\frac{\rho_{si}K_{\rm{R}}}{1+K_{\rm{R}}}} ,\frac{\rho_{si}}{1+K_{\rm{R}}}  \right)$, where $ \rho_{si} $ is the self-interference channel strength, and $K_R$ is the Rician coefficient.
The overall system performance is then averaged over 100 channel realizations.
During our simulations, the following values are used to define the default setup: $|\mathbb{K}|=4$ , $ K_{\rm{R}_{\rm{BS}}} = 10 $ $ N_{\rm{BS}} = 32 $, $\rho =\rho_{sr}=\rho_{rd} = -10 dB$, $\rho_{sd}=-40 dB $, $ \rho_{si} = 1 $, $ \sigma_n^2 =(\sigma_{n_r}^k)^2 =(\sigma_{n_d}^k)^2 = -40 dB $, $ \sigma_{e}^2 =(\sigma_{e,sr}^k)^2 =(\sigma_{e,rd}^k)^2 =(\sigma_{e,sd}^k)^2= (\sigma_{e,rr}^k)^2 = -50 dB $, $ P_S = P_R = 1 $, $ \kappa = \beta =-30dB $ where $ {\mathbf{\Theta}_{\rm{t},S}} = \kappa \mathbf{I}_{N} $. 

 %$\mu_i = 0.9 $, $ P_{i_{\rm{zero}}} + P_{i_{\rm{FD}}} = \frac{ P_i}{10}  $ for $ i \in \tilde{\mathbb{N}} $

\begin{figure}[htbp]
\centerline{\includegraphics[width=\columnwidth]{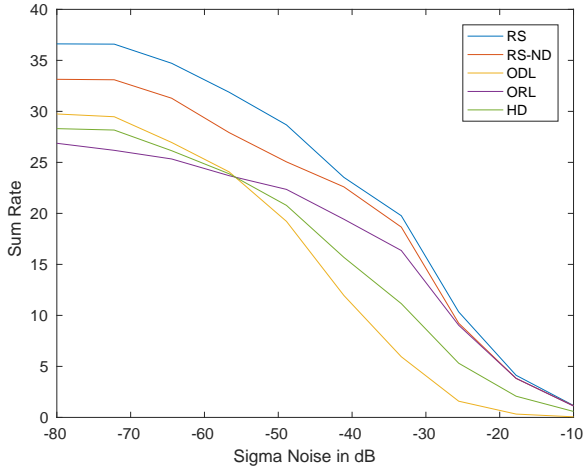}}
\caption{Sum rate  vs. noise}
\label{fig1}
\end{figure}

\begin{figure}[htbp]
\centerline{\includegraphics[width=\columnwidth]{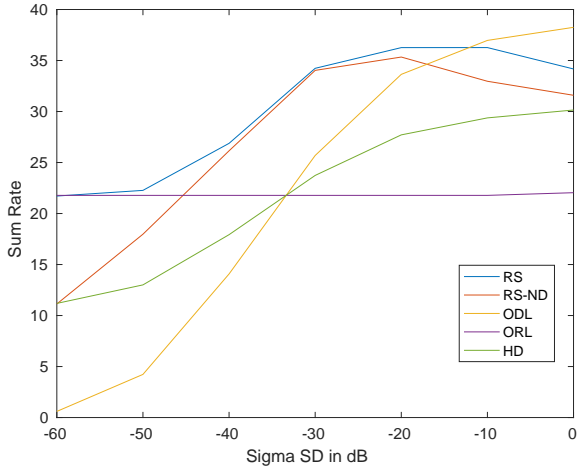}}
\caption{Sum rate  vs. strength of direct channel}
\label{fig2}
\end{figure}

%\begin{figure}[htbp]
%\centerline{\includegraphics[width=\columnwidth]{kappa_cc.eps}}
%\caption{Computational time (CT) vs. hardware inaccuracy}
%\label{fig2}
%\end{figure}

In Fig. \ref{fig2}, the performance of our proposed algorithm  is plotted in terms of total sum rate, for different values of noise at the destination and relay ($\sigma_n^2 $).
As we expect, the sum rate decreases as the noise increases, i.e., lower sum rate for higher noise values.
It is clear that the proposed algorithm outperforms all the other benchmarks. 

In Fig. \ref{fig2}, the performance of the algorithm 1 in terms of system sum rate is evaluated against the strength of the DL.
It can be clearly observed that as the strength of DL increases, the system performance increase.
While for higher $\rho_{sd}$, our proposed algorithm performs worse compared to ODL case.
This is because, destination assumes that the signal received from the relay link is strong, and tries to decode the relay link first in successive interference cancellation method with strong interference from DL.

%% file: Conclusion.tex
In this paper, we addressed a joint sub-carrier and power allocation problem for a DF system, where a full duplex (FD) massive multiple-input-multiple-output (mMIMO) multi-carrier (MC) base station (BS) node  communicates with a FD MC single antenna node using direct link as well as relay link utilizing the RS approach.
The destination employs a successive interference cancellation approach to decode the received signals.
For modelling the system, the impact of hardware distortions resulting in residual self-interference and ICL, and also imperfect CSI are considered. An iterative optimization method is proposed, to solve joint sub-carrier and power allocation problem to maximize the total sum rate maximization, which follows successive inner approximation (SIA) framework to reach the convergence point that satisfies the KKT conditions. Numerical results show that our rate splitting approach performs better compared to non-rate splitting, and half duplex schemes, especially for high SNR scenarios as well as for weak direct link.